\documentclass[letterpaper, 10 pt, conference]{ieeeconf}  

\IEEEoverridecommandlockouts                              
\usepackage{amsfonts,amsmath,amssymb}
\usepackage{graphicx}

\usepackage{ifthen}
\usepackage{color}
\usepackage{cite}

\usepackage{color}
\usepackage{graphics}
\usepackage{algorithmicx}
\usepackage{algpseudocode}
\usepackage[boxruled,vlined,linesnumbered,onelanguage]{algorithm2e}
\usepackage{caption}
\usepackage{epsfig}
\usepackage{chngcntr}
\usepackage{dashrule}
\usepackage[position=b]{subcaption}
\usepackage[hyperfootnotes=false]{hyperref}
\usepackage{listings}
\usepackage{color}
\usepackage{authblk}

\usepackage[multiple]{footmisc}

\usepackage{xspace}
\usepackage{tikz}
\usepackage{subcaption}
\usetikzlibrary{calc, positioning, fit, arrows, arrows.meta, shapes.geometric, shapes.symbols, shapes.misc, patterns}

\newcommand{\R}{{\rm I\!R}}

\def\scr#1{{\cal #1}}
\newtheorem{theorem}{Theorem}
\newtheorem{lemma}{Lemma}

\newtheorem{corollary}{Corollary}

\newtheorem{remark}{Remark}

\newtheorem{definition}{Definition}

\begin{document}

\author{Yuke Li,\thanks{This work was supported by National Science Foundation grant n.1607101.00 and US Air Force grant n. FA9550-16-1-0290. The authors thank an anonymous reviewer, whose comment prompted us to generalize the preliminary ideas in the first submission. Those ideas were also presented in the 2018 MTNS meeting.} and A. Stephen Morse\thanks{Y. Li and A. Stephen Morse are with the Department of Electrical Engineering, Yale University, New Haven, CT, USA, \{yuke.li, as.morse\}@yale.edu}}

\title{The Power Allocation Game on Dynamic Networks: Subgame Perfection}
\maketitle

\begin{abstract}
In the game theory literature, there appears to be little research on equilibrium selection for normal-form games with an infinite strategy space and discontinuous utility functions. Moreover, many existing selection methods are not applicable to games involving both cooperative and noncooperative scenarios (e.g., ``games on signed graphs''). With the purpose of equilibrium selection, the power allocation game developed in \cite{allocation}, which is a static, resource allocation game on signed graphs, will be reformulated into an extensive form. Results about the subgame perfect Nash equilibria in the extensive-form game will be given. This appears to be the first time that subgame perfection based on time-varying graphs is used for equilibrium selection in network games. This idea of subgame perfection proposed in the paper may be extrapolated to other network games, which will be illustrated with a simple example of congestion games.

\keywords subgame perfect Nash equilibrium, equilibrium selection, time-varying graph, extensive form, network games, games on signed graphs

\end{abstract}

\section{Introduction}

In \cite{allocation}, a power allocation game (``PAG'') is developed as a static, distributed resource allocation game on a network of countries (equivalently, agents or decision makers of those countries) connected to each other as friends or adversaries. Pure strategy Nash equilibrium classes are defined for the purpose of making game predictions. For instance, a country may survive in one equilibrium class but not in another \cite{survival}. 

An important question that has apparently not been addressed is equilibrium selection, e.g., whether there are justifiable grounds for the agents to play a certain equilibrium or equilibrium class exclusively. For instance, can they always play the kinds of equilibria in which a certain country survives? Technically speaking, this is an equilibrium selection problem for an $N$-player normal form game with an infinite strategy space. The literature of equilibrium selection methods for finite games (e.g., \cite{kohlberg1986strategic,hillas1990definition,van1991stability,myerson1978refinements,selten1975reexamination}) is thus irrelevant to the purpose of this paper.  


One noteworthy paper on equilibrium selection for infinite normal form games is \cite{simon1995equilibrium}. This work aims to generalize approaches such as Selten's trembling hand perfection criterion originally developed for selecting equilibria in finite games to the case of infinite normal form games with continuous utility functions. However, the methods in \cite{simon1995equilibrium} cannot be applied to the PAG because by the two preference axioms in \cite{survival}, the real-valued utility function representation of countries' preference for the power allocation matrices must be discontinuous \cite{paradox2}. The more recent work in \cite{carbonell2013approximation,carbonell2011existence,carbonell2011perfect} contains the development of equilibrium selection methods for infinite normal-form games with discontinuous utility functions. But the assumptions for the utility functions in those discontinuous games to hold include, for instance, ``payoff-security'', and are therefore inapplicable, either.

Some other existing criteria for equilibrium selection can be shown to be unsuitable for equilibrium selection in the PAG for more substantive reasons. Based on these criteria, the equilibria may be required to be also ``Pareto optimal'', ``coalition-proof'', and so forth, thereby refining the Nash equilibria set. However, the definition of Pareto optimality states that no agent can be better off without making at least another agent worse off. This concept does not seem appropriate for conflictual scenarios where players adopt a certain ``winner-takes-all'' logic. In addition, by the concept of coalition-proof Nash equilibria, any subset of players in the game cannot all strictly benefit with any form of joint deviations. But in the context of the PAG, the idea of joint deviations by a ``coalition'' consisting of some countries and their adversaries may be farfetched. 

In this paper we apply the notion of time-varying graphs to the PAG, i.e., to have the PAG take place in a sequence of changing networked environments, and explore the subgame perfect Nash equilibria in the extensive-form PAG. The first investigations of extensive form games include \cite{zermelo1913anwendung},\cite{selten1978chain} and 
\cite{selten1965spieltheoretische}; one classic application is the Stackelberg game for the study of market competition, where a leader acts first before the followers choose to whether to enter the market to compete with it. In the context of the PAG, in an ``ascending chain'' sequence, the environment may gradually ``expand'' to incorporate more adversary and friend relations among countries over time. In a ``descending chain'' sequence, the environment may undergo ``faults'' over time in the sense that some relations may disappear. In any of the sequences, countries optimize by sequentially allocating their power to a queue of friends and adversaries. A single decision maker's optimization will be done much like in other resource allocation problems such as packing. However, a major difference from these other problems is that in a resource allocation game, it will be of more interest to explore multiple decision makers' allocation strategies which are in best responses with one another. The idea of subgame perfection based on time-varying graphs has seldom been used in equilibrium selection for network games. It seems that this idea may also be applicable to other network games, such as congestion games (to which subgame perfection based on agents' sequential moves like in the Stackelberg game has been applied, e.g.,  \cite{milchtaich1998crowding}).

This paper is organized as follows.  First the power allocation game formulated in \cite{allocation} will be briefly
summarized in Section II. Then in Section III the extensive- form PAG will be described and finally in Section IV results pertaining to the new game form will be stated. In Section V a simple example of how the method can be applicable to congestion games will be illustrated. Lastly, modifications to the presented idea appropriate to different contexts will be discussed.

\section{Review: The PAG in Normal Form} 

\subsection{Basic Idea} By the  {\em power allocation game} or PAG is meant a distributed resource allocation game between $n$ countries with labels in $\mathbf{n}= \{1,2,\ldots,n\}$. The game is formulated on a simple, undirected,  signed graph $\mathbb{G}$ called ``an environment graph'' \cite{survival} whose $n$ vertices correspond to the countries and whose $m$ edges represent relationships between countries. An edge between distinct vertices $i$ and $j$, denoted by $(i,j)$, is labeled with a plus sign if countries $i$ an $j$ are friends and with a minus sign if countries $i$ and $j$ are adversaries.  Let the set of all friendly pairs be $\mathcal{R}_{\mathcal{F}}$ and the set of all adversarial pairs be $\mathcal{R}_{\mathcal{A}}$.  
For each $i\in\mathbf{n}$, $\scr{F}_i$ and $\scr{A}_i$ denote the sets of labels of country $i$'s friends and adversaries respectively; it is assumed that $i\in\scr{F}_i$ and that $\scr{F}_i$ and $\scr{A}_i$ are disjoint sets.
Each country $i$ possesses a nonnegative quantity $p_i$ called the {\em total power} of country $i$. An allocation of this power or {\em strategy} is a nonnegative $n\times 1$ row vector $u_i$ whose $j$ component $u_{ij}$ is that part of $p_i$ which country $i$ allocates under the strategy to either support country $j$ if $j\in\scr{F}_i$ or to demise country $j$ if $j\in\scr{A}_i$; accordingly $u_ij= 0$ if $j\not\in\scr{F}_i\cup\scr{A}_i$ and $u_{i1}+u_{i2}+\cdots +u_{in} = p_i$. The goal of the game is for each country to choose a strategy which contributes to the demise of all of its adversaries and to the support of all of its friends.

Each set of country strategies $\{u_i,\;i\in\mathbf{n}\}$ determines an
 $n\times n$ matrix  $U$ whose $i$th row is $u_i$. Thus $U = [u_{ij}]_{n\times n}$ is a nonnegative matrix such that, for each $i\in\mathbf{n}$, $u_{i1}+u_{i2}+\cdots +u_{in} = p_i$. Any such matrix is called a {\em strategy matrix}  and $\scr{U}$ is the set of all $n\times n$ strategy matrices.

\subsection{Multi-front Pursuit of Survival}

In \cite{allocation} and \cite{survival}, how countries allocate the power in the support of \emph{the survival} of its friends and the demise of that of its adversaries is studied, which is in line with the fundamental assumptions about countries' behavior in classical international relations theory.\cite{waltz1979theory} The following additional formulations are offered: 

Each strategy matrix $U$ determines for each $i\in\mathbf{n}$, the {\em total support} $\sigma_i(U)$ of
 country $i$ and the {\em total threat}  $\tau_i(U)$ against country $i$.  Here
 $\sigma_i:\scr{U}\rightarrow\R$ and $\tau_i:\scr{U}\rightarrow \R$ are non-negative
 valued maps defined by
$U\longmapsto \sum_{j\in\scr{F}_i}u_{ji} +\sum_{j\in\scr{A}_i}u_{ij}$
 and $U\longmapsto \sum_{j\in\scr{A}_i}u_{ji}$
  respectively. Thus country $i$'s total support is the sum of the amounts of power
   each of country $i$'s friends
  allocate to its support plus the sum of the amounts of power country $i$ allocates  to the
  destruction of  all of its adversaries. Country $i$'s total threat, on the other hand, is the sum of the amounts of power
  country $i$'s adversaries allocate to its destruction. These allocations in turn determine
   country $i$'s {\em state} $x_i(U)$ which may be safe, precarious, or unsafe depending on the relative
    values of $\sigma _i(U)$ and $\tau_i(U)$. In particular, 
     $x_i(U) = $ safe
    if $\sigma_i(U)>\tau_i(U)$, $x_i(U)=$ precarious if $\sigma_i(U)=\tau_i(U)$, or
    $x_i(U) = $ unsafe if $\sigma_i(U)<\tau_i(U)$.

In playing the PAG, countries select individual strategies in accordance with certain weak and/or strong preferences.
A sufficient set of conditions for  country $i$ to {\em weakly prefer} strategy
matrix $V\in\scr{U}$ over strategy matrix $U\in\scr{U}$ are as follows
\begin{enumerate}
\item  For all $j\in\scr{F}_i$
 either $x_j(V)\in$ \{safe, precarious\}, or $x_j(U)\in$ \{unsafe\}, or both.
 \item For all $j\in\scr{A}_i$ either $x_j(V)\in $ \{unsafe, precarious\}, or $x_j(U)\in$ \{safe\}, or both.
\end{enumerate}
Weak preference  by country $i$ of  $V$ over $U$ is denoted by $U \preceq V$.

Meanwhile, a sufficient condition for country
 $i$ to be {\em indifferent} to the choice between $V$ and $U$ is  that $x_i(U)=x_j(V)$ for all $j\in\scr{F}_i\cup\scr{A}_i$.  This is denoted by $V\sim U$.

Finally, a sufficient condition for country $i$  to {\em strongly prefer} $V$ over $U$ is that $x_i(V)$ be a  safe or precarious  state and $x_i(U)$ be an unsafe state. Strong preference by country $i$ of $V$ over $U$ is denoted by $U \prec V$.

\section{The PAG in Extensive Form}

\subsection{Sequence of Spanning Subgraphs}

Let $\mathbb{G} = (\mathcal{V}, \mathcal{E})$ be called an ``environment graph'' as in \cite{survival}. Write $\mathcal{G}$ for the set of all spanning subgraphs of $\mathbb{G}$. Let a sequence of spanning subgraphs $\mathbb{G}(t),\;t \in \{0, 1, 2, \ldots, n\}$ from $\mathcal{G}$ be such that $\mathbb{G}(t) \in \mathcal{G}$, $t \in \{0, 1, 2, \ldots, n\}$. Let $\scr{F}_{i}(t)$ and $\scr{A}_{i}(t)$ respectively be the sets of labels of country $i$'s friends and adversaries at time $t$. 

A sequence of spanning subgraphs $\mathbb{G}(t),\;t \in \{0, 1, 2, \ldots, n\}$ from $\mathcal{G}$ is an {\em ascending chain} if $\mathbb{G}(t)\subset \mathbb{G}(t+1),\;t \in \{0, 1, 2, \ldots, n\}$ where by $\mathbb{G}(t)\subset \mathbb{G}(t+1)$ we mean that the edge set of $\mathbb{G}(t)$ is contained in the edge set of $\mathbb{G}(t+1)$. Conversely, a sequence of spanning subgraphs $\mathbb{G}(t),\; t \in \{0, 1, 2, \ldots, n\}$ from $\mathbb{G}$ is a {\em descending chain} if $\mathbb{G}(t+1)\subset \mathbb{G}(t),\;t \in \{0, 1, 2, \ldots, n\}$.

\begin{remark} Let the ascending sequence $\mathbb{G}(t),\;t \in \{0, 1, 2, \ldots, n\}$  {\em reach} $\mathbb{G}$ from $\mathbb{G}(0)$ in $n$ steps, i.e., $\mathbb{G}(n) = \mathbb{G}$. For such as an ascending chain of spanning subgraphs that reaches the environment subgraph $\mathbb{G}$ at time $t$, the number of subgraphs at time $t$ has to satisfy,
$$\sum_{\beta = 0}^{m - \alpha} \frac{(m-\alpha)!}{\beta!}$$ where $m$ is the total number of edges in $\mathbb{G}$, $\alpha$ is the number of edges in $\mathbb{G}(t-1)$ and $\beta$ is the number of edges in $\mathbb{G}(t)$.
\end{remark}


\subsection{Sequence of Decisions}

 At time $t \in \{0, 1, 2, \ldots, n\}$, every country $i$ decides on its strategy $u_{i}(t)$, i.e., the amount of resources being allocated to its friends and adversaries labeled respectively in $\mathcal{F}_{i}(t)$ and $\mathcal{A}_{i}(t)$, subject to its total power constraint $i\in\mathbf{n}$, $u_{i1}(t)+u_{i2}(t)+\cdots +u_{in}(t) = p_i$.  Various decision rules may be assumed, which will be discussed in Section IV. For time $t \in \{1, 2, \ldots, n\}$, let the set of power allocation matrices at layer $t$ be represented as $\mathcal{U}(t) \subset \mathbb{R}^{n \times n}$ where $\forall U(t) \in \mathcal{U}(t)$ and $i \in \mathbf{n}$.



The information structure of the dynamic game is complete information. When making each possible allocation at time $t$, each country has observed \emph{the power allocation path} prior to time $t$, which is $$U(0), U(1), \ldots, U(t-1).$$ At the end of the sequence, each country $i$ receives its state $x_{i}(U(t),\; t \in \{0, 1, 2, \ldots, n\})$ as the outcome of \emph{the power allocation path} from $t = 0$ to $t=n$, $$U(0), U(1), \ldots, U(n).$$  In other words, as consistent with a standard extensive-form game, the power allocation outcome is only realized at $t = n$.

\subsection{Sequence of Subgames}


The power allocation game in extensive form can be represented by a decision tree $\mathbb{T}$ with $n$ layers and a nonempty set of decision nodes at layer $t$ where $t \in \{0, 1, 2, \ldots, n\}$. As opposed to other extensive form games (e.g., two-player games), at each decision point, all countries will decide on their allocations on the corresponding subgraph, ending up playing a ``smaller'' version of the original normal form PAG. 

The root node denotes the decision point of the countries involved in environment graph at $t = 0$, $\mathbb{G}(0)$. Each decision node at each layer denotes the point the countries have to decide on the allocations on the friend and adversary relations involved in the environment graph $\mathbb{G}(t)$. From each node at layer $t$, there grows an infinite number of branches, the $q$th of which represents a possible allocation strategy $U_{q}(t)$ made by countries to those friends and adversaries. The number of the branches between any node at layer $t$ and its successors at layer $t+1$ is the cardinality of $\mathcal{U}(t)$.

In addition, each decision node in the tree $T$ represents an information set. As is commonly defined, an information set is a set of decision nodes that establishes all the possible allocations that could have taken place in the game so far, given what the players that will act next have observed. Assuming complete and perfect information (i.e., the power allocation path leading to the particular decision node has already been observed by countries), each information set in the tree is a singleton. 

\begin{figure}
\centering
\includegraphics[width=0.5\textwidth]{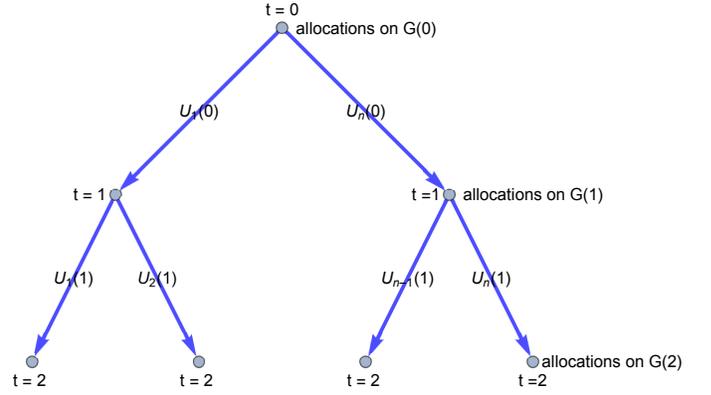}
\caption{An example of a decision tree T up to t = 2}
\end{figure}

Such a decision tree defines a sequence of subgames. In this extensive form game framework, the $q$-th ($q \in \mathbb{N}$) decision node at layer $t$ of $T$ ($t \in \{0, 1, 2, \ldots, n\}$) and all its successors make up a subgame at layer $t$; let the set of subgames at layer $t$ be $\mathcal{\kappa}(t)$. Obviously, the total number of decision nodes in $T$ equals the total number of subgames. Each path in the tree $T$ represents a power allocation path from $t = 0$ to $t = n$, $U(0), U(1), \ldots, U(n)$. A function $\eta:$ $$\mathcal{U}(0) \times \mathcal{U}(1) \ldots \times \mathcal{U}(n) \longrightarrow \mathcal{\kappa}(0) \times \mathcal{\kappa}(1) \ldots \times \mathcal{\kappa}(n)$$ maps a power allocation path to a sequence of $n+1$ subgames it has traversed, where the $t$-th subgame of this sequence can be represented as $\eta(U(0), U(1), U(n))_{t}$, $t \in \{0, 1, 2, \ldots, n\}$. As consistent with the assumptions in the normal form, there is only a finite number of possible power allocation outcomes realized at the terminal nodes of the tree, i.e., a number of $3^n$ possible state vectors which countries will partially order based on the axioms in the setup of the PAG.

\subsection{Subgame Perfect Nash Equilibrium}

In an extensive-form game, it is natural to investigate the subgame perfect Nash equilibrium. In extensive form games with finite strategy space, a subgame perfect Nash equilibrium should in principle include the complete plan of every agent's action in every instance they would encounter; however, it is impossible to do the same for the PAG, which motivates the below definition based on equilibrium paths of power allocation strategy matrices.

\begin{definition}[Subgame Perfection Nash Equilibrium]
A power allocation path $$U^*(0), \ldots U^*(t) \ldots U^*(n)$$ on a sequence of the spanning subgraphs of the environment graph $\mathbb{G}$ is a subgame perfect Nash equilibrium for the PAG $\Gamma$ in extensive form if and only if it is an equilibrium in all of the $n+1$ subgames it traverses. 

\end{definition}

\subsection{Extension with Incomplete Information}

Though not the focus of this paper, a straightforward extension can be made to incorporate incomplete information into the extensive-form PAG. Suppose there exists an additional agent called the nature. The nature draws a probability distribution over the space of the sequences. Let a probability distribution over the measurable space $\Omega$ of all possible sequences of the spanning subgraphs of $\mathbb{G}$ be $\delta$, and where the $h$-th sequence takes the probability mass $\delta_{h}$. On each possible sequence, countries play the corresponding extensive-form game as formulated previously. 
 
\begin{remark} Denote a subgame perfect equilibrium of a PAG assuming a possible sequence $h$ as $(U(0), U(1), \ldots, U(n))_{h}$. Assuming all countries have a ``common prior'' about the probability distribution $\delta$ over $\Omega$, $(U(0), U(1), \ldots, U(n))_{h}$ is always a Bayesian Nash equilibrium (see relevant definitions in \cite{harsanyi1967games}). 

\end{remark}

\section{Main Results}

When analyzing subgame perfect Nash equilibrium of the extensive-form PAG, it is necessary to assume a specific decision rule for the game. Two specific examples are below, where the second can be regarded as a special case of the first.


\subsubsection*{Decision Rule 1}
At time $t \in \{1, 2, \ldots, n\}$, every country $i$ collects its allocations to its friends and adversaries labeled in $\mathcal{F}_{i}(t-1)$ and $\mathcal{A}_{i}(t-1)$ back into its reserved power $u_{ii}(t)$, and decides how to reallocate $u_{ii}(t)$, which is equal to $p_{i}$, to its friends and adversaries labeled in $\mathcal{F}_{i}(t)$ and $\mathcal{A}_{i}(t)$.

\subsubsection*{Decision Rule 1.1} At time $t \in \{1, 2, \ldots, n\}$, every country $i$ keeps constant its allocations to its friends and adversaries labeled in $\mathcal{F}_{i}(t-1)$ and $\mathcal{A}_{i}(t-1)$ that have not disappeared at $t$, collects allocations to those that have disappeared at $t$ back into $u_{ii}(t)$, and allocates its reserved power $u_{ii}(t)$ to its new friends and adversaries labeled in $\mathcal{F}_{i}(t) -  \mathcal{F}_{i}(t-1)$ and $\mathcal{A}_{i}(t) - \mathcal{A}_{i}(t-1)$.

\begin{lemma} In an extensive-form PAG on a sequence of spanning subgraphs of $\mathbb{G}$, the subgame perfect Nash equilibrium set by assuming Decision rule 1 is the superset of subgame perfect Nash equilibrium set by assuming Decision rule 2. 
\end{lemma}

\begin{proof}
Trivial by the definition of Decision Rule 1.1. Under Decision Rule 1, when reallocating, let every country allocate the same amount of power as it does at $t-1$ to its friends and adversaries that still remain at $t$, and the rest of the power to new friends and adversaries labeled in $\mathcal{F}_{i}(t) -  \mathcal{F}_{i}(t-1)$ and $\mathcal{A}_{i}(t) - \mathcal{A}_{i}(t-1)$. 

Therefore, any subgame perfect Nash equilibrium $U^*(0), U^*(1), \ldots, U^*(n)$ by assuming Decision rule 1.1 must also be an equilibrium path by assuming Decision rule 1.

\end{proof}

\begin{theorem}
Given a pure strategy Nash equilibrium $U^*$ of the normal-form PAG on $\mathbb{G}$, if in $U^*$ there exists two countries $i$ and $j$ who have made zero allocations to each other $u^*_{ij} = u^*_{ji} = 0$, $U^*$ can be realized on the equilibrium path of a subgame perfect Nash equilibrium in an extensive-form PAG on a sequence of spanning graphs of 
$\mathbb{G}$. This holds independently of the decision rule assumed. 
\end{theorem}

\begin{proof} Let an ascending chain of two environment graphs $\mathbb{G}(0)$, $\mathbb{G}(1)$ be such that $\mathbb{G}(1) = \mathbb{G}$, and $\mathbb{G}(0) \subset \mathbb{G}_{1}$ in the sense that $\mathcal{E}_{1} - \mathcal{E}_{0} = (i,j)$. 

Proceeding backwardly, a power allocation path $U(0), U(1)$ is mapped to a sequence of two subgames $\kappa_{0}$ and $\kappa_{1}$. $\kappa_{0}$ is the subgame where $U_{1}$ is chosen, and $\kappa_{1}$ is the subgame where the path $U(0), U(1)$ is chosen, which can be regarded as the extensive-form game itself. 

Let $U(1) = U(0) = U^*$. At $t = 1$, $U_{1}$ is by definition pure strategy Nash equilibrium in $\kappa_{0}$. At $t = 0$, none of the agents would like to deviate from the strategies in $U(0)$. Therefore, the path $U(0), U(1)$ is a subgame perfect Nash equilibrium.

\end{proof}

\begin{theorem} Given any sequence of spanning subgraphs of the environment graph of the normal-form PAG $\mathbb{G}$, a subgame perfect Nash equilibrium can always be found in the extensive-form PAG on that sequence by assuming Decision Rule 1.

\end{theorem}

\begin{proof}By Decision Rule 1, at time $t \in \{0, 1, \ldots, n\}$, suppose countries are playing the normal-form PAG on $\mathbb{G}(n)$. Since it has been shown that any normal-form PAG has a pure straetgy Nash equiilbrium, let $U^*(n)$ be a pure strategy Nash equilibrium of this game. 

$U^*(0), U^*(1), \ldots, U^*(n)$ is an optimal path for the countries in  the extensive game and obviously a subgame perfect Nash equilibrium.

\end{proof}

\begin{remark}
Given a particular subgame perfect Nash equilibrium, the possible sequences of spanning subgraphs on which this equilibrium is realized in the corresponding extensive-form game always exists (by definition) but may not be unique. For example, in Lemma 1, there exists an opposite, descending chain of the spanning subgraphs $\mathbb{G}(1)$, $\mathbb{G}(0)$ where $U(1), U(0) = U^*$ is the subgame perfect Nash equilibrium. 

\end{remark}


\begin{theorem}

A balanced equilibrium $U^*$ of the normal-form PAG on $\mathbb{G}$ as defined in \cite{balancing} can be realized on the equilibrium path of a subgame prefect Nash equilibrium in an extensive-form PAG on a sequence of spanning subgraphs of $\mathbb{G}$ by assuming Decision Rule 1.

\end{theorem}

\begin{proof} A power allocation matrix $U$ of a PAG is a balanced equilibrium if
 
 \begin{enumerate}
\item $\forall i \in \mathbf{n}$ such that $ \mathcal{A}_{i}$ is the empty set, $u_{ii} = p_{i}$.\label{c1}
\item $\forall i \in \mathbf{n}$ such that $\mathcal{A}_{i} $ is nonempty,
 $u_{ii} = 0$ and $$\sum_{j \in \mathcal{A}_{i}}u_{ij} = p_{i}$$ \label{c2}
\item $\forall (i,j) \in \mathcal{R}_{\mathcal{A}}$, $u_{ij} = u_{ji}$.
\end{enumerate}

Case I: $\mathcal{R}_{\mathcal{A}} \neq \emptyset$. Take an adversary pair $(i,j) \in \mathcal{R}_{\mathcal{A}}$. Let an ascending chain of spanning subgraphs of $\mathbb{G}$, $\mathbb{G}(0)$, $\mathbb{G}(1)$ be such that $\mathbb{G}(1) = \mathbb{G}$, $\mathbb{G}(0) \subset \mathbb{G}(1)$ in the sense that $\mathcal{E}_{1} - \mathcal{E}_{0} = (i,j)$. 

Assuming Decision Rule 1, let $U(1) = U^*$ where $U^*$ is a balanced equilibrium. Let $U(0)$ be only different from $U(1)$ in that the allocations $u_{ij}(0) = u_{ji}(0) = 0$, and $u_{ii}(0) = u_{jj}(0) = u_{ij}(1) = u_{ji}(1).$

A power allocation path $U(0), U(1)$ is thus mapped to a sequence of two subgames $\kappa_{0}$ and $\kappa_{1}$. $\kappa_{0}$ is the subgame where $U_{1}$ is chosen, and $\kappa_{1}$ is the subgame where the path $U(0), U(1)$ is chosen. 

$U_{0}$ is a pure strategy Nash equilibrium of $\kappa_{0}$. $U(0), U(1)$ constitutes a subgame perfect Nash equilibrium in the extensive form game asssuming the above graph sequence.

Case II: $\mathcal{R}_{\mathcal{A}} = \emptyset$. There always exists a pure strategy Nash equilibrium of the normal-form PAG where two countries $i$ and $j$ such that $u^*_{ij} = u^*_{ji} = 0$. 

By Theorem 1, $U^*$ can always be realized on the equilibrium path of a subgame perfect Nash equilibrium in an extensive form of the PAG.

\end{proof}

{\it Example 1:} Let the environment graph of the normal-form PAG $\mathbb{G}$ is a Petersen graph. Assume that each of the ten countries has two adversaries and two friends (itself and one external) as shown in Fig. 2(b). Assume in addition the ascending chain of spanning subgraphs $\mathbb{G}(0), \mathbb{G}(1)$, where  $\mathbb{G}(0)$ contains two separate cycles of adversary relations and $\mathbb{G}(1) = \mathbb{G}$. As illustrated in the allocation graphs in Fig. 3 whose definition is in \cite{survival}, $U(1)$ is a balanced equilibrium and the power allocationpath $U(0), U(1)$ is a subgame perfect Nash equilibrium.

\begin{figure}

    \begin{subfigure}{0.45\textwidth}
        \centering
    \includegraphics[width=0.8\linewidth]{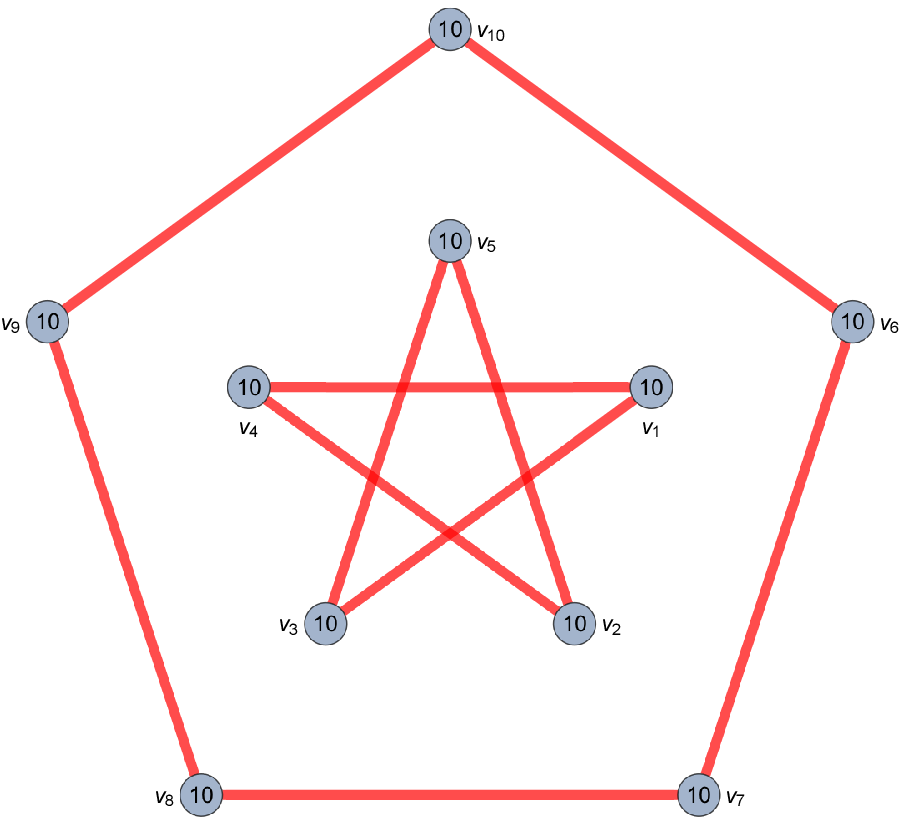}
    \caption{$\mathbb{G}(0)$}
    \label{fig:my_label}
    \end{subfigure}
    \begin{subfigure}{0.45\textwidth}
        \centering
    \includegraphics[width=0.8\linewidth]{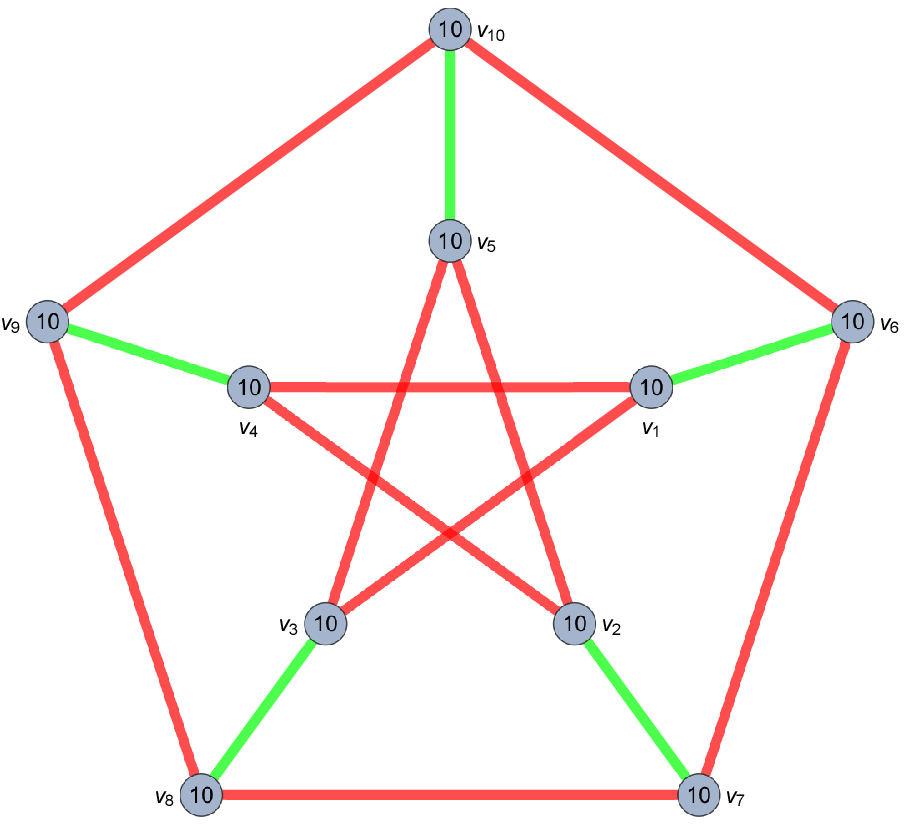}
    \caption{$\mathbb{G}(1)$}
    \label{fig:my_label}
    \end{subfigure}
    \caption{An ascending chain of spanning subgraphs}
\end{figure}

\begin{figure}
    \begin{subfigure}{0.45\textwidth}
        \centering
    \includegraphics[width=0.77\linewidth]{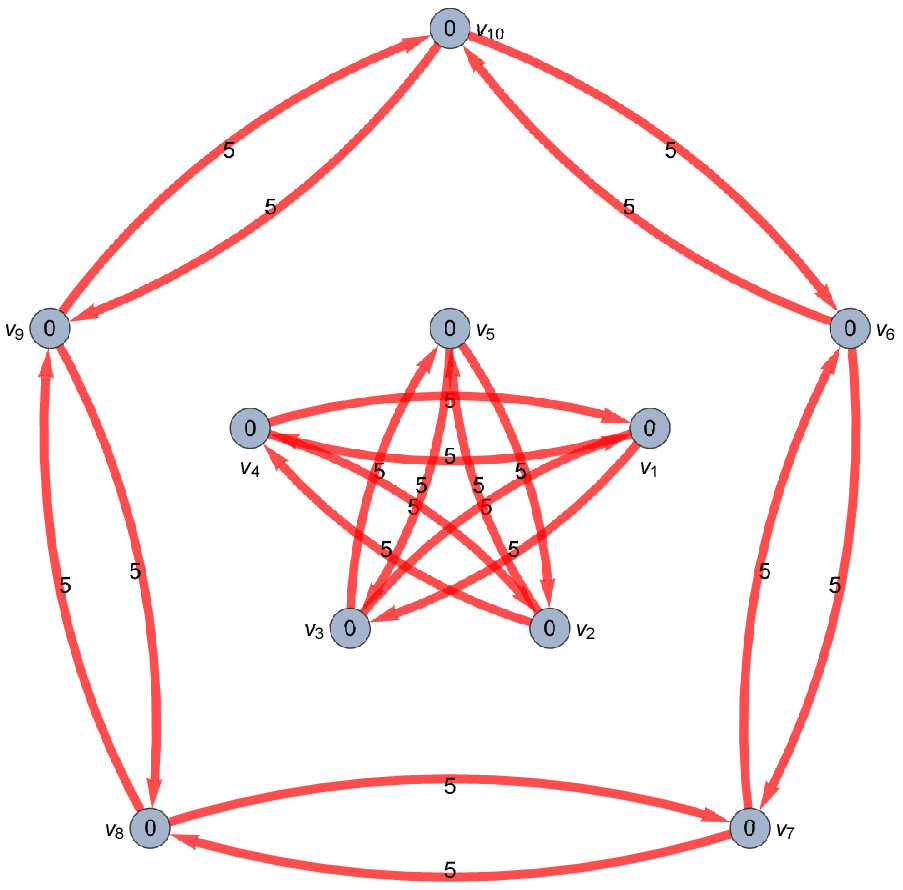}
    \caption{$U(0)$}
    \label{fig:my_label}
    \end{subfigure}
    \begin{subfigure}{0.45\textwidth}
        \centering
    \includegraphics[width=0.77\linewidth]{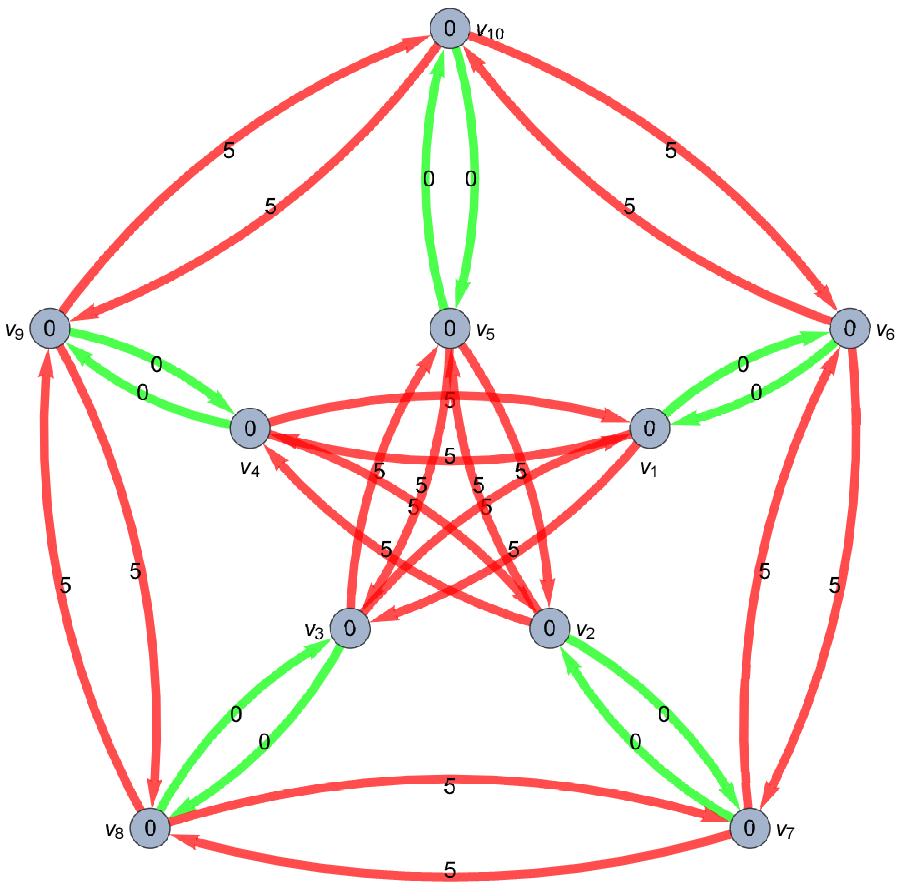}
    \caption{$U(1)$}
    \label{fig:my_label}
    \end{subfigure}
    \caption{Allocations on the ascending chain}
\end{figure}

\begin{theorem}
 The constructive algorithm in \cite{allocation} derives a subgame perfect Nash equilibrium in an extensive-form PAG on a sequence of spanning subgraphs of $\mathbb{G}$ by assuming Decision Rule 1.1. 

\end{theorem}
\begin{proof}First, let  $q \in \mathbb{N}$ be the number of adversarial pairs in $\mathcal{R}_{\mathcal{A}}$, and $\mathbf{m} = \{1, 2, \dots, q\}$ be the set of labels. Let $\gamma: \mathcal{R}_{\mathcal{A}} \mapsto \mathbf{m}$. 
At time $t$, $r^{-1}(t) = \{i,j\}$ is the $t$-th adversary pair being traversed by the algorithm.  Let the remaining power of countries $i$ and $j$ involved in this relation be $z_{i}(t)$ and $z_{j}(t)$, $t \in \{1, 2, \ldots, q\}$. 

Given the environment graph of the PAG $\mathbb{G}$, the algorithm in \cite{allocation} proceeds with traversing each adversary relation and constructing allocations on the relation consecutively. This actually gives rise to an ascending chain of environment graphs that reaches the $\mathbb{G}$ at the last step, $\mathcal{E}(0) = \mathcal{R}_{\mathcal{F}}, \mathcal{E}(1) = \mathcal{R}_{\mathcal{F}} \cup r^{-1}(1), \ldots, \mathcal{E}(q+1) = \mathcal{R}_{\mathcal{F}} \cup r^{-1}(1) \cup r^{-1}(2) \cup \ldots ...\cup r^{-1}(q)$. This algorithm is consistent with Decision Rule 1.1.

Let a sequence of power allocation matrices $U(0), U(1), \ldots, U(q)$ represent the sequence of allocations constructed by the algorithm where $U(0) = \text{diagonal}\{z_{1}(0), z_{2}(0), \ldots, z_{n}(0)\} =  \text{diagonal}\{p_{1}, p_{2}, \ldots, p_{n}\},~\text{and}~U(t+1) =  \text{diagonal}\{z_{1}(t), z_{2}(t), z_{i}(t) - \min\{z_{i}(t), z_{j}(t)\}, \ldots,  z_{j}(t) - \min\{z_{i}(t), z_{j}(t)\}, z_{n}(t)\}  + \sum_{r^{-1}(t) \in \mathcal{R}_{\mathcal{A}}, t \in \mathbf{q}}(e_{i}^{T}e_{j} + e_{j}^{T}e_{i}) \min\{z_{i}(t), z_{j}(t)\}, t \in \{0, 1, \ldots, q\}$.

A sequence of $q+1$ subgames is derived. Proceeding backwardly, $\kappa_{t}$ is the subgame where the path $U(q - t +1),U(q-t+2), \ldots,U(q)$ is chosen, $t \in \{0, 1, \ldots, q+1\}$. 

In this extensive form game, no country $i$ will want to deviate from its strategy $u_{i}(t)$, $t \in \{0, 1, \ldots, n\}$. Therefore, the sequence of allocations represent a subgame perfect Nash equilibrium. 

\end{proof}

\begin{corollary}
Given a normal-form PAG on $\mathbb{G}$, if for $i \in \mathbf{n}$, its power is no larger than the total power of its two adversaries $p_{i} \leq \sum_{j \in \mathcal{A}_{i}}p_{j}$, there always exists a subgame perfect Nash euilibrium in an extensive-form PAG on a sequence of spanning subgraphs of $\mathbb{G}$ that predicts $i$ to be precarious at the last step $n$ of the power allocation path, that is, $\sigma_{i}(U(n)) = \tau_{i}(U(n))$, by assuming Decision Rule 1.1.

\end{corollary}

\begin{proof}  Let $q \in \mathbb{N}$ be the number of adversarial pairs in $\mathcal{R}_{\mathcal{A}}$, and $\mathbf{m} = \{1, 2, \dots, q\}$ be the set of labels. Suppose $i$ has $g$ adversaries ($ g \leq q$, and $g \in \mathbb{N}$). Without loss of generality, assume that the set of labels for these adversaries is $\mathbf{m'} = \{1, 2, \ldots, g\} \subset \mathbf{m}$. As in Theorem 3, let $\gamma: \mathcal{R}_{\mathcal{A}} \mapsto \mathbf{m}$, and assume the same algorithm in \cite{allocation}.

From $t = 0$ to $t = g$, the algorithm proceeds with each adversary relation of $i$ and constructing allocations on the relation consecutively. From $t = g+1$ to $t = q$, the algorithm will traverse the rest of the adversarial pairs. 

A sequence of power allocation matrices $U(0), U(1), \ldots, U(q)$ constructed by the algorithm is $U(0) = \text{diagonal}\{z_{1}(0), z_{2}(0), \ldots, z_{n}(0)\} =  \text{diagonal}\{p_{1}, p_{2}, \ldots, p_{n}\},~\text{and}~ 
U(t+1) =  \text{diagonal}\{z_{1}(t), z_{2}(t), z_{i}(t) - \min\{z_{i}(t), z_{j}(t)\}, \ldots,  z_{j}(t) - \min\{z_{i}(t), z_{j}(t)\}, z_{n}(t)\}  + \sum_{r^{-1}(t) \in \mathcal{R}_{\mathcal{A}}, t \in \mathbf{q}}(e_{i}^{T}e_{j} + e_{j}^{T}e_{i}) \min\{z_{i}(t), z_{j}(t)\}, t \in \{0, 1, \ldots, q\}$. 

This gives rise to an ascending chain of environment graphs that reaches the $\mathbb{G}$ at the last step, $\mathcal{E}(0) = \mathcal{R}_{\mathcal{F}}, \mathcal{E}(1) = \mathcal{R}_{\mathcal{F}} \cup r^{-1}(1), \ldots, \mathcal{E}(g) = \mathcal{R}_{\mathcal{F}} \cup r^{-1}(1) \cup r^{-1}(2) \cup \ldots ...\cup r^{-1}(g), \ldots, \mathcal{E}(q+1) = \mathcal{R}_{\mathcal{F}} \cup r^{-1}(1) \cup r^{-1}(2) \cup \ldots ...\cup r^{-1}(q)$.


For $t \in \{m, m+1, \ldots, q\}$, $z_{i}(m) = 0$. By Theorem 3, $U(0), U(1), \ldots, U(n)$ is a subgame perfect Nash equilibrium. Therefore, for $t \in \{m, m+1, \ldots, q\}$, $\sigma_{i}(U(t)) = \tau_{i}(U(t))$.

\end{proof}

\begin{theorem}

Given a normal-form PAG on a complete graph of only adversary relations, if countries' power condition satisfies $p_{i} \leq \sum_{j \in \mathcal{A}_{i}}p_{j}$,  a subgame perfect Nash equilibrium in which only a country is safe, that is, $\sigma_{i}(U(n)) > \tau_{i}(U(n))$ at the last step $n$ of the power allocation path, can be guaranteed through a class of sequences of spanning subgraphs assuming Decision Rule 1.1. 

\end{theorem}

\begin{proof}
 Given an arbitrary country $i$ in $\mathbf{n}$, the set of adversarial pairs except for those involving $i$ is denoted as $\mathcal{R}_{\mathcal{A}} - \{\{i, j\}: j \in \mathcal{A}_{i}\}$. Note that $\mathcal{R}_{\mathcal{A}} - \{\{i, j\}: j \in \mathcal{A}_{i}\}$ still make up a complete subgraph of $\mathbb{G}$, $\mathbb{G}' = \{\mathbf{n}- \{i\}, \mathcal{E}'\}$. 
 
 Let an ascending sequence of two spanning subgraphs of $\mathbb{G}$ be such that $\mathbb{G}(0) = \mathbb{G}'$ and $\mathbb{G}(1) = \mathbb{G}$.

\begin{enumerate}
\item If there exists a country $j$ in the subgraph $\mathbb{G}(0)$, that is, $j \in \mathbf{n} - \{i\}$, such that its power is no smaller than that of all other countries (i.e., its adversaries) combined in the subgraph, $$p_{j} > \sum_{k \in \mathcal{A}_{j} - \{i\}}p_{k}.$$  

First, on $\mathbb{G}(0)$, let country $j$ allocate enough to make all of its adversaries other than $i$ unsafe. Construct an $U(0) = [u_{jk}(0)]_{(n-1) \times (n-1)}$ where there holds $$\forall k \in \mathcal{A}_{j}- \{i\}, u_{jk}(0) > p_{k}$$ and $$\sum_{h \in \mathcal{A}_{k}- \{i\}}u_{kh}(0) = p_{k}$$

Then on $\mathbb{G}(0)$, assuming Decision Rule 1.1, let country $i$ allocate enough to make $j$ unsafe. Construct an $U(1) = [u_{ij}]_{n \times n}$ by expanding $U(0)$ to incorporate the allocations between $i$ and countries in $\mathbf{n}- \{i\}$. Let $u_{ij}(1) > p_{j} - \sum_{k \in \mathcal{A}_{j}- \{i\}}u_{jk}(0)$. 

This is feasible because, as assumed, $p_{i} < \sum_{j \in \mathcal{A}_{i}}$.

Then $$p_{j} - p_{i} \leq \sum_{k \in \mathcal{A}_{j}- \{i\}}p_{k} \leq \sum_{k \in \mathcal{A}_{j}- \{i\}}u_{jk}(0)$$

Rearranging terms, $$p_{i} \geq p_{j} - \sum_{k \in \mathcal{A}_{j}- \{i\}}u_{jk}(0).$$

None of the countries would like to deviate from its strategies in $U(0), U(1)$. Therefore, a subgame perfect Nash equilibrium has been derived such that $\sigma_{i}(U(1))  > \tau_{i}(U(1))$ and
$\sigma_{j}(U(1)) < \tau_{j}(U(1))$ for all $j \in \mathbf{n}- \{i\}.$

\item If there does not exist a country in $\mathbb{G}(0)$ such that its power exceeds all other countries in the subgraph.  By \cite{balancing}, a balancing equilibrium $U'$ exists for the PAG of the $n-1$ countries on $\mathbb{G}'$. 

Let it be $U(0) = [u_{jk}(0)]_{(n-1) \times (n-1)}$, where by definition $$\forall j \in \mathbf{n}', u_{jj}(0) = 0; \forall j, k \in \mathbf{n}', u_{jk}(0) = u_{jk}(0);$$ $$\sum_{k \in \mathcal{A}_{j}- \{i\}} u_{jk}(0) = p_{j}.$$

In this case, as consistent with Decision Rule 1.1, expand $U(0)$ to incorporate the allocations between $i$ and countries in $\mathbf{n}- \{i\}$ to obtain $U(1)$. $\forall j \in \mathbf{n}- \{i\}$, let  $u_{ij}(0) = \frac{p_{i}}{n-1}$. 

None of the countries would like to deviate from its strategies in $U(0), U(1)$. Then a subgame perfect Nash equilibrium has been derived such that $\sigma_{i_0}(U(1)) > \tau_{i_0}(U(1))$ and
$\sigma_{j}(U(1)) < \tau_{j}(U(1))$ for all $j \in \mathbf{n}- \{i\}.$

\end{enumerate}

\end{proof}

\begin{remark}

An important implication from Corollary 1 and Theorem 5 is that the particular state $i$ can possibly be in a PAG becomes controllable if the some particular conditions hold, thus fulfilling the purpose of equilibrium selection. The control is specifically done by a sequence of spanning subgraphs of the environment graph as well as a suitably defined decision rule. 

\end{remark}

\section{Example: Application to Other Network Games}

Consider a two-step ascending chain of a road network where at $t = 0$ only road A is opened up and only at $t = 1$ road B is added. Let each agent's action space be $(A, B)$. Similarly, a descending chain where at $t = 0$ the two roads A and B both exist and at $t = 1$ one road, A, gets closed down and thus disappears from the network.

Assume the following decision rule: at each step, each agent chooses whether to take the road in the subgraph. Once an agent has chosen a road, it cannot take any more roads later.  At $t = 1$, the payoffs as shown in the above table are realized. 

In the ascending chain case, the only subgame perfect Nash equilibrium of this extensive form game is $(A, B)$ because agent 1 will make sure to take road A at $t = 0$ and agent 2 takes road B at $t = 1$. Neither will have incentives to deviate. However, when the opposite holds, i.e., when at $t = 0$ only road B is opened up and only at $t = 1$ road A is added, $(B, A)$ will instead be the only subgame perfect Nash equilibrium.

\section{Conclusion}

This paper has focused specifically on selecting (or refining) the pure strategy Nash equilibria set of the normal-form PAG based on the extensive-form game. Two assumptions are invoked in the formulation of this problem. First, a state vector is realized only at the end of a power allocation path as the power allocation outcome. Second,  the graph at each time step is a spanning subgraph of the environment graph in the normal-form PAG. 

The ideas in the paper can be adapted to the studies of other problems of interest. For this to happen, the two assumptions should be relaxed. In terms of a differential game problem where countries would optimally control their power allocation paths in changing environments, the problem formulation should have a state or payoff vector realized at the end of every period on the allocation path. Countries' power and relations may also change over time, either deterministically or stochastically. Moreover, the environment graph at each time instant may be any possible signed graph of those countries. In terms of equilibrium selection in network formation games, it should be noted that agents' strategies are to form or change the network itself in these games. Accordingly instead of assuming environment graphs that change over time, the kinds of agents that are able to choose their strategies should be assumed to change instead. For instance, at each time step, a subset of agents is only allowed to choose their strategies of whether to connect with others, where a simple leader-follower sequence as mentioned before will be a special case.



\bibliographystyle{unsrt}

\end{document}